\documentclass[12pt]{article}

\usepackage{amsmath}
\usepackage{amssymb}
\usepackage{amsthm}
\usepackage{mathrsfs}

\usepackage{slashed}
\usepackage{color}
\usepackage{esint}

\usepackage{graphicx}
\usepackage[all,cmtip]{xy}
\usepackage{hyperref}

\def\cutint{{\int \!\!\!\!\!\! -}}
\newtheorem{thm}{Theorem}

\newtheorem{prop}{Proposition}
\newtheorem{cor}{Corollary}

\newtheorem{defn}{Definition}

\hypersetup{
    pdfborder = {0 0 0},%
    colorlinks,%
    citecolor=blue,%
    filecolor=black,%
    linkcolor=red,%
    urlcolor=green
}

\begin{document}

\title{ Noncommutative  topological $\mathbb{Z}_2$ invariant}
\author{ Ralph M. Kaufmann \thanks{e--mail: rkaufman@math.purdue.edu  }
		\and
		Dan Li  \thanks{e--mail: li1863@math.purdue.edu } \\
       \and
        Birgit  Wehefritz-Kaufmann   \thanks{e--mail: ebkaufma@math.purdue.edu     }
	\\			\\			
        Keywords: noncommutative geometry, topological $\mathbb{Z}_2$ index,\\
	       fixed point algebra, KQ-cycle, NC Kane--Mele invariant
	    }

\date{}
\maketitle

\begin{abstract}
We generalize the $\mathbb{Z}_2$ invariant of topological insulators using noncommutative differential geometry in two different ways.
First, we model Majorana zero modes by KQ-cycles in the framework of analytic K-homology, and we define the noncommutative  $\mathbb{Z}_2$ invariant as a topological index
in noncommutative topology.
Second, we look at the geometric picture of the Pfaffian formalism of the $\mathbb{Z}_2$ invariant, i.e., the Kane--Mele invariant,
and we define the noncommutative Kane--Mele invariant  over the fixed point algebra of the time reversal symmetry
 in the noncommutative 2-torus. Finally, we are able to prove the equivalence between the noncommutative topological $\mathbb{Z}_2$ index and the noncommutative Kane--Mele invariant.

\end{abstract}

\section{Introduction}\label{Intro}

This work is inspired by the index theory of topological insulators \cite{KLW1501}, 
and our results will be applied to disordered topological insulators \cite{P11}. 
Because of the time reversal  $\mathbb{Z}_2$ symmetry, there exists a $\mathbb{Z}_2$-valued topological invariant 
characterizing time reversal invariant topological insulators. 
There are many equivalent characterizations of the topological $\mathbb{Z}_2$ invariant from different perspectives \cite{KLW15}.
In particular, the topological $\mathbb{Z}_2$ invariant can be understood in the framework of index theory and K-theory.
More precisely, the topological $\mathbb{Z}_2$ invariant can be interpreted as a mod 2 index theorem  \cite{KLW1501}.  
In this framework, there are basically three different ways to compute the topological $\mathbb{Z}_2$ invariant,  
the mod 2 spectral flow as the analytic index, the topological index and its exponentiated version, i.e., the  Kane--Mele invariant \cite{KM05}.
We will generalize the topological index and the Kane--Mele invariant of time reversal invariant topological insulators onto  noncommutative manifolds in this paper. 

First of all, the topological band theory of a topological insulator is defined by the Bloch bundle $\pi: \mathcal{B} \rightarrow X $, which is a finite rank Hilbert bundle 
over the momentum space $X$ \cite{KLW15, KLW1501}. For simplicity, $X$ is assumed to be a compact space without boundary.  
The topological index can be computed by the integral of the Chern character of the Bloch bundle over $X$.
Physically, the Chern character defines the effective classical field theory of a topological insulator. 
When the time reversal symmetry is taken into account, the Bloch bundle is equipped with a real structure and becomes a Quaternionic vector bundle. 
In other words, the Bloch bundle defines an element in the Quaternionic K-theory of $X$, i.e., $KQ^*(X)$ \cite{DG15}.
In a modern language, the topological index can be computed by the pairing between $KQ$-theory and $KQ$-homology \cite{KLW1501}.
In practice, it is convenient to use $KR$-theory and $KR$-homology, since there exists a canonical isomorphism between $KQ$-theory and $KR$-theory.

When the base manifold  is replaced by a noncommutative manifold defined by a $C^*$-algebra such as the noncommutative 2-torus, 
the topological index can be defined with the help of the Connes--Chern characters of K-homology and K-theory. 
By the machinery of noncommutative topology, the topological index can be computed
by pairing periodic cyclic cohomology and homology after applying the Connes--Chern characters. 
Furthermore, the local index formula of Connes and Moscovici gives another way to compute the topological index based on residue traces and spectral zeta functions \cite{CM95}.

The Kane--Mele invariant is the exponentiated version of the topological index \cite{KLW1501}, 
which gives the effective quantum field theory of a topological insulator.
In the classical case, the topological index and the Kane--Mele invariant are equivalent \cite{FM13, KLW1501}.
The Kane--Mele invariant is defined by the product of signs of Pfaffians over the fixed points \cite{FKM07, KM05}.
If we interpret the theory of a topological insulator as a topological quantum field theory (TQFT), then the Kane--Mele
invariant defines the partition function  of a fully extended TQFT \cite{KLW16}. Moreover, this extended TQFT is
completely determined by the fixed points. In the bulk-boundary correspondence of a topological 
insulator, the Kane--Mele invariant suggests that the boundary theory is given by the fixed points, which can be
reformulated by a family index theorem.

The geometric picture of the Kane--Mele invariant is more interesting for us, 
which can be interpreted as the comparison of orientations between determinant and Pfaffian line bundles. 
Based on this geometric picture, we will generalize the Kane--Mele invariant to the noncommutative 2-torus, where
the set of fixed points is replaced by the fixed point $C^*$-algebra.  The Pfaffian line bundle over the fixed point algebra
will be constructed similarly as the determinant line bundle over the noncommutative 2-torus. Finally, 
the noncommutative Kane--Mele invariant will be defined by the product of local orientations of the Pfaffian line bundle over the fixed point algebra.

This paper is organized as follows. In section 2, we will first review the index theory of Majorana zero modes and interpret the
$\mathbb{Z}_2$ invariant as a mod 2 topological index. Next we will apply the machinery of noncommutative topology to
define the noncommutative $\mathbb{Z}_2$ invariant as a topological index.  In section 3, we will first review the Kane--Mele invariant and reformulate
it with the help of determinant and Pfaffian line bundles. 
Next we will define the noncommutative Kane--Mele invariant over the fixed point algebra based on Quillen's construction
of the determinant line bundle over Fredholm operators. Finally, we will establish the equivalence between those two noncommutative versions of
the topological $\mathbb{Z}_2$ invariant.

\section{Topological index} \label{TopInd}
Let us first review the topological $\mathbb{Z}_2$ invariant as a mod 2 index theorem,
for more background and details see \cite{KLW15, KLW1501}. More precisely, the topological $\mathbb{Z}_2$ invariant counts the parity of localized Majorana zero modes, i.e.,
a mod 2 spectral flow, which can be computed by a mod 2 topological index. This interpretation of the topological $\mathbb{Z}_2$ invariant
can be directly generalized   as a topological index in noncommutative topology.

\subsection{Mod 2 index theorem}
In a time reversal invariant topological insulator, the fundamental objects to study are Majorana zero modes, and
their collective effect gives rise to the topological $\mathbb{Z}_2$ invariant. We have to point out that Majorana zero modes are different
from Majorana fermions, since Majorana zero modes are quasi-particles rather than real particles. In the framework of index theory and K-theory,
we use a vector bundle to model a topological insulator, and a K-cycle to model Majorana zero modes. 
On the one hand, we  model the band structure of a topological insulator by
a vector bundle, then we consider all such vector bundles and compute  K-theory. 
On the other hand, we model Majorana zero modes by a K-cycle in K-homology. And then we can pair K-homology with K-theory to get an index number.  

\subsubsection{KQ-theory}
Let us first review the K-theoretic approach to obtain global features of the band structure of a topological insulator.
First let $X$ be a compact space without boundary representing the momentum space of a topological insulator.
Time reversal symmetry defines an involution on $X$, i.e., a homeomorphism $\tau: X \rightarrow X$
such that $\tau^2 = id_X$, so $(X, \tau)$ is an involutive space or a Real space with the real structure $\tau$.
With the time reversal transformation $\tau$, $X$ has the structure of a $\mathbb{Z}_2$-CW complex.  
Let $X^\tau$ denote the fixed points of $\tau$, 
\begin{equation*} 
   X^\tau := \{ {x} \in X ~|~ \tau( {x} ) =  {x} \}
\end{equation*}
which is assumed to be a finite set. Then $X$ can be built up from the fixed points $X^\tau$ by
gluing $k$-cells ($1 \leq k \leq \dim X$) that carry a free $\mathbb{Z}_2$ action, i.e., $\mathbb{Z}_2$-equivariant $k$-cells.

The band structure of a topological insulator defines a complex vector bundle over the momentum space $\pi: \mathcal{B} \rightarrow X$, 
which is called the Bloch bundle by physicists.
The Bloch bundle models the finite dimensional Hilbert space of physical states, so it is a Hilbert bundle, i.e., each fiber is a Hilbert space. 
Time reversal symmetry defines a time reversal operator $\Theta$, which is an anti-unitary operator acting on electronic states.
Then it induces an anti-linear bundle isomorphism $\Theta: \mathcal{B} \rightarrow \mathcal{B}$ such that $\Theta^2 = - id_\mathcal{B}$.  
In other words, $\Theta$ defines a general real structure on $\mathcal{B}$, and it is sometimes called an anti-involution.
Taking  the real structures $\tau$ and $\Theta$ into account, one has the Hilbert bundle 
$\pi: (\mathcal{B}, \Theta) \rightarrow (X, \tau)$ with $\tau^2 = 1$ and $\Theta^2 = -1$, which is called a Quaternionic vector bundle. 
 
The Quaternionic K-group $KQ(X, \tau)$ 
is the Grothendieck group of finite rank Quaternionic vector bundles over  $(X,\tau)$.
When the involution $\tau$ is understood, one simply uses the notation $KQ(X) = KQ(X, \tau) $. 
Similar to the complex K-theory,
higher $KQ$-groups  can be defined by suspensions,
 $KQ$-theory can be extended to locally compact spaces,  and 
the reduced  $KQ$-group  $\widetilde{KQ}(X)$ is defined  by  the kernel of the restriction map
$ i^*: KQ(X) \rightarrow KQ(pt)$
induced from the inclusion $ i: {pt} \hookrightarrow X $.
There exists a canonical isomorphism 
$$
KQ^*(X) \cong KR^{*-4}(X)
$$ 
so it is convenient to compute $KQ$-groups by  $KR$-groups \cite{A66}. 
Furthermore, the Bott periodicity of KQ-theory can be derived from that of KR-theory, i.e., $KQ^{8 -i}(X) \cong KQ^{-i}(X)$.

Therefore, all possible band structures of a topological insulator can be classified by computing KQ-theory. 
For example, for a 2d momentum space $X$, the topological $\mathbb{Z}_2$ invariant belongs to $KQ(X)$. 
While for a 3d $X$, the topological $\mathbb{Z}_2$ invariant belongs to the odd KQ-group $KQ^{-1}(X)$. Although the 
topological $\mathbb{Z}_2$ invariant could also belong to $KQ(X)$ in 3 dimensions,  index theory tells us that 
the odd topological index must live in an odd KQ-group, and that is $KQ^{-1}(X)$ for a 3d $X$. 

\subsubsection{Analytic index}
If one looks into the local geometry  behind the global KQ-theory, one will see a mod 2 spectral flow of localized Majorana zero modes.
 Next let us recall the definition of 
Majorana zero modes and the analytic index.

One considers a free electronic system in a topological insulator,  let  $H$ be a Hamiltonian acting on electronic states that respects the time reversal symmetry. 
By perturbation theory, $H$ takes different forms around a fixed point and a regular point. So
the Hamiltonian $H$ is usually assumed to be parametrized by the momentum space $X$, that is, $H$ is given by a family of Hamiltonians $H(x)$ for $x \in X$. 
$H$ being time reversal invariant means $H(x)$ satisfies the following equation with respect to the time reversal transformation $\tau$ and time reversal operator $\Theta$, 
$$
\Theta H(x) \Theta^* = H(\tau(x)), \quad \forall\,\, x \in X
$$
When $H(x)$ is applied to an electronic state $\psi(x)$, one has the eigenvalue equation,
$$
H(x)\psi(x) = E(x) \psi(x)
$$
where $E(x)$ is the energy function.  At the same time, one also has a similar equation for $\Theta \psi$,
$$
 \Theta H(x) \Theta^* (\Theta \psi (x)) = \Theta [E(x) \psi(x) ] = E(\tau(x)) \Theta \psi(x)
$$
which is equivalent to 
$$
H(\tau(x))\psi(\tau(x)) = E(\tau(x)) \psi(\tau(x))
$$
As a consequence, the pair $(\psi, \Theta \psi)$ has the same energy level. We have to point out that $\psi$ and $\Theta\psi$ in general have different domains $U $ and $\tau(U)$,
if $\psi$ is a local state defined over an open subset $U\subset X$. 

Now we define the effective Hamiltonian $\tilde{H}$  by
$$
\tilde{H} := \begin{pmatrix}
             0 & \Theta H(x) \Theta^* \\
             H(x) & 0
            \end{pmatrix} = 
            \begin{pmatrix}
             0 &  H(\tau(x)) \\
             H(x) & 0
            \end{pmatrix} 
$$
acting on a pair  $\Psi = (\psi, \Theta \psi)$ so that 
$$
\tilde{H} \Psi = \begin{pmatrix}
             0 & \Theta H(x) \Theta^* \\
             H(x) & 0
            \end{pmatrix} 
            \begin{pmatrix}
             \psi \\
             \Theta \psi
            \end{pmatrix} =  \begin{pmatrix}
            \Theta E \Theta\psi \\
             E\psi
            \end{pmatrix}
$$
where $\Theta E(x) = E(\tau(x))$. Such a pair  $\Psi = (\psi, \Theta \psi)$, which consists of
an electronic state and its mirror partner under the time reversal symmetry, is called a Majorana state since
$\Psi$ satisfies a real condition, see later.

\begin{defn}
   Majorana zero modes are Majorana states $\Psi_0 = (\psi_0, \Theta \psi_0)$ that are zero modes of the effective Hamiltonian $\tilde{H}$, i.e., $\tilde {H} \Psi_0 = 0$.
\end{defn}
Because of the definition of $\tau$ and $\Theta$, the local states $\psi$ and $\Theta \psi$ have overlap domains only around some fixed point, so 
a localized Majorana zero mode is a pair $(\psi_0, \Theta \psi_0)$ such that $\psi_0(x) = 0$ and $\Theta \psi_0(x) = 0$ at some $x \in X^\tau$.
For example, in 3d a localized Majorana zero mode has the local geometry of a double cone, 
that is, in a small neighborhood of a fixed point $o \in X^\tau$, it can be described by the equation
$\{ (x, y,z) ~|~x^2 + y^2 = z^2\}$ and $o = (0,0,0)$.

\begin{defn}
    The analytic index of the effective Hamiltonian $\tilde{H}$ is defined as the parity of Majorana zero modes,
\begin{equation}
   ind_a= ind_2(\tilde{H}) = \dim \ker \tilde{H} \quad \text{mod 2}
\end{equation}
\end{defn}
It is $\mathbb{Z}_2$-valued since $KO^{-1}(pt) = \mathbb{Z}_2$. More precisely,
$\tilde{H}$ can be viewed as a skew-adjoint Fredholm operator  \cite{KLW1501}, and a classifying space of $KR^{-1}$
is given by skew-adjoint Fredholm operators \cite{AS69}.

The $\mathbb{Z}_2$ invariant of a topological insulator, denoted by $\nu$, happens to be also defined as the parity of Majorana zero modes, i.e., the analytic index,
$$
\nu = ind_a 
$$
 In a localized Majorana zero mode $(\psi_0, \Theta \psi_0)$, if we 
call either $\psi_0$ or $\Theta \psi_0$ a chiral zero mode, then the spectral flow of  such a chiral zero mode can be used to compute the analytic index.
As a result, the analytic index of $\tilde{H}$ can be computed by the mod 2 spectral flow of $H$ or $\Theta H \Theta^*$, namely,
the parity of chiral zero modes running through all fixed points (at most one chiral zero mode for each fixed point). 

\subsubsection{Topological index}
By the Atiyah--Singer index theorem, the analytic index can be computed by the topological index.
In our case, the topological index is the integral of the Chern character of the Bloch bundle $\mathcal{B}$ over the momentum space $X$.
In addition, the topological index being $\mathbb{Z}_2$ valued is an implication from the time reversal symmetry, which can be
proved both locally and globally. 

For example, when $X = \mathbb{T}^2$, the topological $\mathbb{Z}_2$ invariant comes from $\widetilde{KQ}(\mathbb{T}^2) = \mathbb{Z}_2$.
In this case, the topological $\mathbb{Z}_2$ invariant can be computed by the  topological index, which is the Chern character of the projection $p$,
\begin{equation}
  ind(p) = \frac{1}{2\pi}\int_{\mathbb{T}^2} tr(pdpdp)
\end{equation}
where $p$ is the representative of the generator of $\widetilde{KQ}(\mathbb{T}^2) \cong \mathbb{Z}_2$. 
More precisely, the projection $p$ represents the Bloch bundle $\pi: (\mathcal{B}, \Theta) \rightarrow (X, \tau)$.
This index $ind(p)$ is also called the first Chern number, which is naturally $\mathbb{Z}_2$-valued since the first Chern class can be proved
to be a 2-torsion \cite{KLW1501}.

When $X = \mathbb{T}^3$, the topological $\mathbb{Z}_2$ invariant as an index comes from $\mathbb{Z}_2 \in \widetilde{KQ}^{-1}(\mathbb{T}^3) $.
In this case, the $\mathbb{Z}_2$ invariant can be computed by the odd topological index, which is the odd Chern character of the unitary $g$, 
\begin{equation}
   ind(g)  = \frac{1}{4\pi^2 }   \int_{\mathbb{T}^3} tr(g^{-1}dg)^{3}
\end{equation}
where $g$ is the representative of the generator of $\mathbb{Z}_2 \in \widetilde{KQ}^{-1}(\mathbb{T}^3) $.
More precisely, $ g: (\mathbb{T}^3, \tau) \rightarrow (U(2), \sigma) $ is the transition function of 
the Quaternionic Hilbert bundle $\pi: (\mathcal{B}, \Theta) \rightarrow (X, \tau)$,
where $\sigma$ is an involution defined on the structure group $U(2)$ by $\sigma(g) \mapsto -g^{-1}$ such that $\sigma^2 = 1$. 
The index $ind(g)$ being $\mathbb{Z}_2$-valued can be proved locally based on the compatibility condition between the involutions 
$\tau$ and $\sigma$ \cite{KLW1501}.

\subsubsection{KQ-homology}
Now we introduce the KQ-cycle of Majorana zero modes, and then the topological index can be computed by pairing KQ-homology  with KQ-theory.
First, we define a new real structure based on the  time reversal operator $\Theta$,
$$
J = \begin{pmatrix}
    0 & \Theta^* \\
    \Theta & 0
    \end{pmatrix}
$$
such that $J^* = J$ and $J^2 = 1$. Since the real structure $J$ acts on a pair $(\psi, \Theta \psi)$ by
$$
\begin{pmatrix}
 0 & \Theta^* \\
 \Theta & 0
\end{pmatrix} \begin{pmatrix} 
                \psi \\
                \Theta \psi
              \end{pmatrix} = \begin{pmatrix}
                                \psi \\
                                \Theta \psi
                              \end{pmatrix}
$$
 $\Psi = (\psi, \Theta \psi)$ satisfies the real condition $J \Psi = \Psi$ and  $\Psi$ is called a Majorana state.
The real structure $J$ can be augmented by tensoring with the identity matrix $J \times I_n$, still call it $J$, so that 
the Hilbert bundle $\pi: (\mathcal{B}, J) \rightarrow (X, \tau)$ becomes a Real vector bundle with $\tau^2= 1$ and $J^2 = 1$.
We stress that each fiber of $\mathcal{B}$ can be viewed as a vector space over quaternions $\mathbb{H}$, whose standard basis is given by $\{ i, \Theta, i\Theta \}$.

We define a generalized $KR_1$-cycle by the quadruple 
\begin{equation}\label{genKR}
   (C(X), L_{\mathbb{H}}^2(X, \mathcal{B}), \tilde{H}, J)
\end{equation}
such that 
$$
J \tilde{H} = - \tilde{H}J, \quad J^2 = 1
$$
where the Hilbert space $L^2_{\mathbb{H}}$ is defined over quaternions $\mathbb{H}$.
We want to emphasize the time reversal operator $\Theta$, so the above quadruple is equivalently written as,  
$$
(C(X), L^2(X, \mathcal{C}) \oplus  L^2(X, \Theta \mathcal{C}) , \begin{pmatrix}
                                                                                        0 & \Theta H \Theta^* \\
                                                                                        H& 0 
                                                                                     \end{pmatrix}
, \begin{pmatrix}
    0 & \Theta^* \\
    \Theta & 0
  \end{pmatrix}
)
$$
where the Bloch bundle is decomposed into $ \mathcal{B} \cong \mathcal{C} \oplus \Theta \mathcal{C}$ separating the chiral states.
If we add the usual grading operator
$\gamma = \begin{pmatrix}
          1 & 0 \\
          0 & -1 
         \end{pmatrix}$
with $\gamma = \gamma^*$ and $\gamma^2 = 1$, then $\gamma$ is used to separate the chiral states,
$$
 \frac{1 + \gamma}{2} \Psi = \begin{pmatrix}
                               1 & 0 \\
                               0 & 0 
                             \end{pmatrix} \begin{pmatrix}
                                             \psi \\
                                             \Theta \psi
                                            \end{pmatrix} = \begin{pmatrix}
                                                               \psi \\
                                                                0
                                                             \end{pmatrix}, \quad 
 \frac{1 - \gamma}{2} \Psi = \begin{pmatrix}
                               0 & 0 \\
                               0 & 1 
                             \end{pmatrix} \begin{pmatrix}
                                             \psi \\
                                             \Theta \psi
                                            \end{pmatrix} = \begin{pmatrix}
                                                               0 \\
                                                                \Theta \psi
                                                             \end{pmatrix}                                                            
$$
\begin{defn}
 The KQ-cycle of Majorana zero modes is defined as the quintuple, 
\begin{equation}\label{KQ6cyc}
   (C(X), L^2(X, \mathcal{C}) \oplus  L^2(X, \Theta \mathcal{C}) , \begin{pmatrix}
                                                                                        0 & \Theta H \Theta^* \\
                                                                                        H& 0 
                                                                                     \end{pmatrix}
, \begin{pmatrix}
    0 & \Theta^* \\
    \Theta & 0
  \end{pmatrix},   \begin{pmatrix}
          1 & 0 \\
          0 & -1 
         \end{pmatrix}
)
\end{equation}
such that 
$$
J \tilde{H} = - \tilde{H}J, \quad J^2 = 1, \quad J\gamma  = - \gamma J
$$
\end{defn}
The above quintuple \eqref{KQ6cyc} is an even $KQ$-cycle, which is basically the same as the generalized $KR_1$-cycle \eqref{genKR},
so \eqref{KQ6cyc} is viewed as a $KR_2$-cycle, i.e., a $KQ_6$-cycle. The equivalence class of a KQ-cycle defines an element in the KQ-homology.
 In real computations, it is more convenient to use KR-cycles and KR-homology rather than KQ-cycles and KQ-homology. 

As a remark, if we define an operator $\hat{H}$ by rewriting $\tilde{H}$,
$$
\hat{H} : = \begin{pmatrix}
          H & 0 \\
          0 & \Theta H \Theta^* 
      \end{pmatrix}
$$
 then the following quintuple defines a $KR_6$-cycle, call it the twisted KQ-cycle, 
\begin{equation}\label{KQ2cyc}
   (C(X), L^2(X, \mathcal{C}) \oplus  L^2(X, \Theta \mathcal{C}) , \begin{pmatrix}
                                                                                        H & 0 \\
                                                                                        0 & \Theta H \Theta^*  
                                                                                     \end{pmatrix}
, \begin{pmatrix}
    0 & \Theta^* \\
    \Theta & 0
  \end{pmatrix},   \begin{pmatrix}
          1 & 0 \\
          0 & -1 
         \end{pmatrix}
)
\end{equation}
such that 
$$
J\hat{H} = \hat{H}J, \quad J^2 = 1, \quad J\gamma  = - \gamma J
$$

\subsubsection{Pairing KQ-homology  with KQ-theory }
\begin{prop}
 For a two dimensional momentum space $X$, the topological index is obtained by the index pairing
 between the $KQ$-cycle \eqref{KQ6cyc} with the Bloch bundle $\pi: (\mathcal{B}, J) \rightarrow (X, \tau)$.
\end{prop}
\begin{proof}
   Let $p$ be the projection representing the Real Hilbert bundle $\pi: (\mathcal{B}, J) \rightarrow (X, \tau)$ over quaternions $\mathbb{H}$.
   Since this is a quaternionic  bundle instead of a complex bundle, 
   its class $[p] \in KR^{-2}(X) = KQ^{-6}(X)$ instead of $KR^{-1}(X)$.
   
   Denote the class of  the $KQ$-cycle \eqref{KQ6cyc} by
   $$
   [(X, \tilde{H}, \gamma)] = [(C(X), L^2(X, \mathcal{B}), \tilde{H}, J, \gamma)] \in KR_2(X) = KQ_6(X) 
   $$
   Therefore, the topological index is given by the pairing between $KR_2(X)$ and $KR^{-2}(X)$,
   $$
   ind(p) = \langle [(X, \tilde{H}, \gamma)], [p] \rangle \in KO^{-2}(pt) = \mathbb{Z}_2
   $$
\end{proof}

\begin{prop}
  For a three dimensional momentum space $X$, the odd topological index is obtained by the index pairing between
  the $KQ$-cycle \eqref{KQ6cyc} with 
  the Bloch bundle $\pi: (\mathcal{B}, J) \rightarrow (X, \tau)$.
\end{prop}
\begin{proof}
 Let  $g: (X, \tau) \rightarrow (Sp(n), \sigma)$ be the unitary representing the transition function of 
  the Bloch bundle $\pi: (\mathcal{B}, J) \rightarrow (X, \tau)$. Since the image of $g$ is in the compact symplectic group $Sp(n)$,
  its class $[g] \in KQ^{-2}(X) = KR^{-6}(X)$ instead of $KQ^{-1}(X)$.
  
  For odd dimensional spaces, as a convention a Dirac operator is always written diagonally in spin geometry, 
  if so, we use the KQ-cycle \eqref{KQ2cyc} instead of 
  the original KQ-cycle \eqref{KQ6cyc}. And its class is viewed as a $KR_6$-cycle,
  $$
  [(X, \hat{H}, \gamma)] = [C(X), L^2(X, \mathcal{B}), \hat{H}, J, \gamma) ] \in KR_6(X) = KQ_2(X)
  $$
  So the pairing between the twisted KQ-cycle \eqref{KQ2cyc} and the transition function $g$ is a pairing between $KR_6(X)$ and $KR^{-6}(X)$, 
  $$
   \langle [(X, \hat{H}, \gamma)], [g] \rangle \in KO^{-6}(pt) 
  $$
  Finally, the topological index is obtained by another ``twist'' back to the original KQ-cycle \eqref{KQ6cyc},
  $$
  ind(g) = \langle [(X, \tilde{H}, \gamma)], [g] \rangle \in KO^{-2}(pt) = \mathbb{Z}_2 
  $$

\end{proof}

In fact, the class of \eqref{genKR} representing Majorana zero modes, denoted by $[(X, \tilde{H}, \gamma)]$, is viewed as an element in the 2nd KR-homology,
$$
[(X, \tilde{H})] = [ (C(X), L_{\mathbb{H}}^2(X, \mathcal{B}), \tilde{H}, J, \gamma)] \in KR_2(X) = KR_2(T^*X)
$$
since the time reversal operator $\Theta$ doubles the dimension, i.e., from $\mathbb{C}$ to $\mathbb{H}$.
By Kasparov's analytical K-homology of $C(X)$, it has an analytical index comes from $KR_2(T^*X)$, where $T^*X$ stands for the cotangent bundle of $X$.

On the other hand, by the construction of the topological index by Atiyah, 
the symbol class of the Hamiltonian $\tilde{H}$ (viewed as a skew-adjoint operator) acting on Majorana zero modes comes from the 2nd KR-theory, 
$$
[\sigma(\tilde{H})] \in KR^{-2}(T^*X)
$$
Therefore, the pairing between KQ-homology  and KQ-theory of Majorana zero modes 
is reduced to the pairing between $KR^{-2}(T^*X)$ and $KR_2(T^*X)$. The same idea applies to the odd dimensional case if we lift the construction 
to the cotangent bundle and consider the analytical and topological indices of the corresponding symbol class.

 \subsection{Noncommutative topology}
 According to  the Chern--Weil theory,  the classical Chern character is a map from K-theory to de-Rham cohomology. 
 Furthermore, the Chern character is an isomorphism when the torsion subgroup of the K-group is annihilated by tensoring with $\mathbb{Q}$ 
 and the de-Rham cohomology has rational coefficients.
 
 In noncommutative topology, one has Connes--Chern character maps from K-theory to cyclic homology, and from K-homology to cyclic cohomology respectively.
 As a result, the topological index can be computed by pairing (periodic) cyclic cohomology with (periodic) cyclic homology instead of pairing K-homology  with  K-theory.
 This part about Connes--Chern characters is closely  following the book \cite{K09} and Connes' original paper \cite{C85}.
 Furthermore, the local index formula by Connes and Moscovici provides another way to compute the topological index by residue traces \cite{CM95, H06}. 
 Finally, we generalize the topological index of Majorana zero modes using the powerful tools from noncommutative topology.
 
 \subsubsection{Connes--Chern character of K-theory}
 Let $A$ be a $C^*$-algebra representing a noncommutative manifold,
 Connes constructed the Connes--Chern character maps from
 K-theory to cyclic homology groups,
 $$
 Ch_{2n} : K_0(A) \rightarrow HC_{2n}(A), \quad [e] \mapsto \frac{1}{n!} Tr(e \otimes e \cdots \otimes e)
 $$
 where $Tr$ is the operator trace, and there are $(2n+1)$-copies of $e$ in the trace.
 $$
 \quad  Ch_{2n+1} : K_1(A) \rightarrow HC_{2n+1}(A), \quad [u] \mapsto Tr(u^{-1}\otimes u \otimes \cdots\otimes u^{-1} \otimes u)
 $$
   there are $(n+1)$-pairs of $u^{-1}$ and $u$ in the trace.
 
 It is also convenient to equivalently define the above Connes--Chern characters by a pairing between cyclic cohomology and K-theory,
 $$
 HC^{2n}(A) \times K_0(A) \rightarrow \mathbb{C}, \quad HC^{2n+1}(A) \times K_1(A) \rightarrow \mathbb{C}
 $$
 Furthermore,  these maps induce a pairing between periodic cyclic cohomology with K-theory,
 $$
 HP^0(A) \times  K_0(A)  \rightarrow \mathbb{C}, \quad HP^1(A) \times  K_1(A)   \rightarrow \mathbb{C}
 $$
 More precisely, for an even $(b, B)$-cocycle $\phi  = (\phi_0, \phi_2, \cdots, \phi_{2k}, \cdots)$
 and an idempotent element $e \in A$, the pairing between $HP^0(A)$ and  $K_0(A)$  is defined by 
 $$
 \langle [\phi], [e] \rangle = \phi_0(e) +\sum_{k = 1}^\infty (-1)^{k} \frac{ (2k)! }{k!} \phi_{2k}(e-\frac{1}{2}, e, \cdots, e)
 $$
 And for an odd $(b, B)$-cocycle $\phi  = (\phi_1, \phi_3, \cdots, \phi_{2k+1}, \cdots )$ 
 and an invertible element $u \in A$, the pairing between $HP^1(A)$ and  $K_1(A)$  is defined by 
 $$
 \langle [\phi], [u] \rangle = \frac{1}{\Gamma(\frac{1}{2})}\sum_{k = 0}^\infty {(-1)^{k+1}} k! \phi_{2k+1}(u^{-1}, u, \cdots, u^{-1}, u)
 $$
 
 \subsubsection{Connes--Chern character of K-homology}
 Let $(H, F)$  be an odd $p$-summable Fredholm module over a $C^*$-algebra $A$ and let $n$ be an integer such that $2n \geq p$.
 The odd Connes--Chern character of $(H,F)$, denoted by $Ch^{2m-1}(H, F)$, is defined by 
 $$
 Ch^{2m-1}(H, F) (a_0, \cdots, a_{2m-1}) = (-1)^m \frac{2 \Gamma(m - \frac{1}{2})}{\Gamma(\frac{1}{2})} Tr(F[F, a_0] \cdots [F, a_{2m-1}])
 $$
  For any integer $m \geq n$, the class of the cyclic cocycle $Ch^{2m-1}(H, F)$ is stabilized in the odd periodic 
 cyclic cohomology group, call it the stable odd Connes--Chern character of $(H,F)$, i.e., 
 $$
 Ch^{1} (H, F) = [Ch^{2m-1}(H, F)] \in HP^1(A), \quad m \geq n
 $$
 In other words, the Connes--Chern character defines a map from K-homology to periodic cyclic cohomology,
 $$
   Ch^{1} : K^1(A) \rightarrow HP^1(A), \quad [(H,F)] \mapsto Ch^{1} (H, F)
 $$
 
 Let $(H, F, \gamma)$ be an even $p$-summable Fredholm module over $A$ and let $n$ be an integer such that $2n+1 \geq p$.
 The even Connes--Chern character of $(H, F, \gamma)$, denoted by $Ch^{2m}(H, F, \gamma)$, is defined by,
 $$
 Ch^{2m}(H, F, \gamma)(a_0, \cdots, a_{2m}) = \frac{(-1)^m m!}{2} Tr(\gamma F[F, a_0] \cdots [F, a_{2m}])
 $$
 For any integer $m \geq n$, the class of the cyclic cocycle $Ch^{2m}(H, F, \gamma)$ is stabilized in the even periodic 
 cyclic cohomology group, call it the stable even Connes--Chern character of $(H, F, \gamma)$, i.e.,
 $$
 Ch^{0} (H, F, \gamma) = [Ch^{2m}(H, F, \gamma)] \in HP^0(A), \quad m \geq n
 $$
 In other words, the Connes--Chern character defines a map from K-homology to periodic cyclic cohomology,
 $$
   Ch^{0} : K^0(A) \rightarrow HP^0(A), \quad [(H,F, \gamma)] \mapsto Ch^{0} (H, F, \gamma)
 $$
 
 As a remark, given a spectral triple $(A, H, D)$, if the Dirac operator $D$ is invertible, then one defines $F = D|D|^{-1}$ as the phase of $D$,
 and the triple $(A, H, F)$ is a Fredholm module. Another commonly used definition of $F$ from a Dirac operator is $F = D(1+ D^2)^{-1/2}$
 if $D$ is not invertible.
 
 \subsubsection{Index pairing}
 
 Let $(H = H^+ \oplus H^-, F, \gamma)$ be an even  Fredholm module over a $C^*$-algebra $A$ and let $e \in A$ be an idempotent.
 Restricting the operator $eFe$ to the subspace $eH^+$,  one obtains a Fredholm operator,
 $$
 F_e^+= eFe|_{eH^+ } : eH^+ \rightarrow eH^-
 $$
 In addition, the Fredholm index of $F^+_e$ can be computed by 
 $$
  index(F^+_e) = \frac{(-1)^{n}}{2}Tr(\gamma F[F,e] [F, e] \cdots [F, e])
 $$
 If we use the notation $de = [F, e]$, then the above Fredholm index is equivalently written as
 $$
 index(F^+_e) = {(-1)^{n}} Tr[\gamma e(de)^{2n+2}] = {(-1)^{n}} Tr[\gamma (edede)^{n+1}]
 $$
 Furthermore, using the pairing between cyclic cohomology and  homology,
 $$
 \langle \,\, , \,\, \rangle: HC^{2n}(A) \times HC_{2n} (A) \rightarrow \mathbb{C}
 $$
 the above Fredholm index can be computed by the pairing,
 \begin{equation}
    index(F^+_e) = \langle Ch^{2n}(H, F, \gamma), Ch_{2n}(e) \rangle 
 \end{equation}
which is further written as
\begin{equation}
    index(F^+_e) = \langle Ch^{0}(H, F, \gamma), Ch_{0}(e) \rangle 
\end{equation}
provided $n$ is in the stable range.

Let $(H, F)$ be an odd  Fredholm module over $A$ and let $u \in A$ be an invertible element. Define  
the (positive) projection operator by,
$$
P = \frac{1+ F}{2}: H \rightarrow H
$$
so that $PuP: PH \rightarrow PH$ is a Fredholm operator.
As a remark, it is also possible to define the (negative) projection by $P = (1-F)/2$.
Now the Fredholm index of $PuP$ can be computed by 
$$
index(PuP) = \frac{(-1)^n}{2^{n}} Tr(F[F, u^{-1}][F, u] \cdots [F, u^{-1}][F, u])
$$
Alternatively, using $du= [F, u]$ and $du^{-1} = [F, u^{-1}]$, 
$$
index(PuP) = \frac{(-1)^n}{2^{n}} Tr[F(du^{-1}du)^n]
$$
Furthermore, using the pairing between cyclic cohomology and  homology,
 $$
 \langle \,\, , \,\, \rangle: HC^{2n-1}(A) \times HC_{2n-1} (A) \rightarrow \mathbb{C}
 $$
 the above Fredholm index can be computed by the pairing,
 \begin{equation}
    index(PuP) = \langle Ch^{2n-1}(H, F), Ch_{2n-1}(u) \rangle 
 \end{equation}
which is further written as
\begin{equation}
    index(PuP) = \langle Ch^{1}(H, F), Ch_{1}(u) \rangle 
\end{equation}
provided $n$ is in the stable range.

The above discussion is usually summarized in the following commutative diagram \cite{K09},
 $$
   \xymatrixcolsep{5pc}\xymatrix{
  K^*(A) \times  K_*(A) \ar@<4ex>[d]^{Ch_*} \ar@<-5ex>[d]^{Ch^*}  \ar[r]^-{index}   & \mathbb{Z}  \ar[d]    \\
      HP^*(A) \times  HP_*(A) \ar[r]  &  \mathbb{C}     }
   $$
In noncommutative topology, the Fredholm index originally computed by pairing K-homology with K-theory can now
be computed by pairing periodic cyclic cohomology with periodic cyclic homology after applying the Connes--Chern characters.

For example, let  $(A, H, D)$ be a 2-summable spectral triple representing the noncommutative Riemannian geometry of $A$. 
One defines the associated Fredholm module by $(A, H, F, \gamma)$, 
where $F = D|D|^{-1}$ and $H = H^+ \oplus H^-$ is $\mathbb{Z}_2$-graded with the grading $\gamma$. 
For the class of a projection  $[p] \in K_0(A)$, its topological index can be computed by
 $$
 ind(p)  = \frac{1}{2 }Tr(\gamma p [F, p][F, p]) = \frac{1}{2 }Tr(\gamma p dpdp)
 $$
  with $da = [F,a]$.
 
 If $(A, H, D)$ is a 3-summable spectral triple, then one has the associated Fredholm module $(A, H, F)$, and one defines 
 the projection $P = (1+F)/2$ as usual. 
 For the class of a unitary  $[u] \in K_1(A)$, the Fredholm index of $PuP$ can be computed by 
 $$
 ind(PuP) = - \frac{1}{8}Tr[F( du^{-1}  du)^3] = -\frac{1}{4}Tr[( u^{-1}  du)^3]
 $$

 \subsubsection{Local index formula}
 The local index formula by Connes and Moscovici gives another way to compute the topological index by residue traces and zeta functions. 
By the Hochschild character theorem, the Connes--Chern characters of K-homology are cohomologous to some Hochschild characters,
 that is, they are the same as classes in cyclic cohomology
 $$
   [Ch_0(H, F, \gamma)] = [Ch_0(H, D, \gamma)], \quad [Ch_1(H, F)] = [Ch_1(H, D)]
 $$
 In other words, the Hochschild character $Ch_*(H, D)$ of a spectral triple $(A, H, D)$ gives the local expression of the topological index.
 
  Let $(A, H, D)$ be a regular spectral triple, 
 then the Hochschild cocycle is defined by
 \begin{equation}
    Ch^{n} (H, D)(a_0, \cdots, a_{n}) = \frac{\Gamma(1+ n/2)}{n \cdot n!} Tr_\omega(\varepsilon a_0[D, a_1]\cdots [D,a_{n}]|{D}|^{-n})  
 \end{equation}
  where $Tr_\omega$ is the Dixmier trace. Here $\varepsilon = 1$ in the odd case, and $\varepsilon = \gamma$ the grading operator in the even case.
 The Dixmier trace can be expressed in terms of residue trace, which is the contribution from the top local spectral invariant. Because of the 
 dimension spectrum, there are also other contributions from higher order terms at other dimension spectrum, which can be computed by residue traces.

 The first interesting local index formula happens in 3 dimensions. If  $({A}, {H}, {D})$ is a regular 3-summable spectral triple   
and $u \in {A}$ is a unitary operator,  then
the Fredholm index of $PuP$  can be computed by pairing $K_1({A})$ with a $(b, B)$-cocycle $(\phi_1, \phi_3)$,
\begin{equation}
 Index(PuP) = \phi_1(u^{-1}, u) - \phi_3(u^{-1}, u, u^{-1}, u)
\end{equation}
With the notation of noncommutative integral or residue trace,
\begin{equation}
   \cutint a = Res_{z=0} \text{Tr}\, (a |{D}|^{-z}) 
\end{equation}
one has the following expressions for $\phi_1$ and $\phi_3$ when the dimension spectrum is simple,
\begin{equation}\label{cycphi1}
 \phi_1(a^0,a^1) = \cutint a^0da^1|{D}|^{-1} - \frac{1}{4}\cutint a^0 \nabla(da^1) |{D}|^{-3}+ \frac{1}{8}\cutint a^0 \nabla^2 (da^1)|{D}|^{-5}
\end{equation}
\begin{equation}\label{cycphi3}
 \phi_3(a^0, a^1, a^2, a^3) = \frac{1}{12}\cutint a^0 da^1da^2da^3 |{D}|^{-3}
\end{equation} 
$da = [{D}, a]$ and $\nabla (a) = [{D}^2, a]$.

 \subsubsection{Noncommutative topological index}
 
 In this section, we will apply the machinery of noncommutative topology to define the noncommutative  $\mathbb{Z}_2$ invariant as a topological index. 
 First of all, we will model Majorana zero modes by KQ-cycles (or generalized KR-cycles) over $C^*$-algebras.  
 Then the topological index will be defined by the pairing between periodic cyclic cohomology and periodic cyclic homology after
 applying the Connes--Chern characters to K-homology and K-theory. In addition,  the local index formula gives another way to define
 the noncommutative $\mathbb{Z}_2$ invariant by local spectral invariants.
 
 Given a complex $C^*$-algebra $A$ representing the function algebra of the underlying noncommutative space, or roughly, 
 $A$ defines a noncommutative space. Let $\mathcal{H}$ be a complex Hilbert space that models the physical Hilbert space of a topological insulator.
 There is a representation of $A$ on the bounded operators of $\mathcal{H}$, i.e., $\pi: A \rightarrow B(\mathcal{H})$, so that $A$ is a noncommutative space of bounded observables on $\mathcal{H}$.
 
 Time reversal symmetry defines the time reversal operator $\Theta$, which is an anti-unitary operator with $\Theta^2 = -1$
 acting on both $\mathcal{H}$ and $A$. With the time reversal operator, $(\mathcal{H}, \Theta)$ can be viewed as a Hilbert space over quaternions $\mathbb{H}$, denoted by $\mathcal{H}_\mathbb{H}$.
 Or equivalently, $\mathcal{H}_\mathbb{H}$ is decomposed as $\mathcal{H}_\mathbb{H} = \mathcal{H}_\mathbb{C} \oplus \Theta \mathcal{H}_\mathbb{C} $ 
 with the quaternionic structure defined by $\Theta$.
 In other words, the pair $(\mathcal{H}, \Theta)$ defines a quaternionic Hilbert space with respect to the quaternionic structure $\Theta$.
 In addition, the pair $(A, \Theta)$ is called a \emph{real} $C^*$-algebra with the general real structure (or Quaternionic structure) 
 $\Theta$ over the quaternionic Hilbert space $(\mathcal{H}, \Theta)$.
 Then one considers the KR-theory (or KQ-theory) of $(A, \Theta)$, denoted by $KR_*(A)$ (or $KQ_*(A)$) \cite{R15}. 
 Inspired by the classical case, it is better to use $KQ_*(A)$ for a topological insulator due to $\Theta^2 = -1$. Therefore, 
 given a projection $p \in A$ such that $p^2= p = p^*$, its class $[p] \in KQ_0(A)$ represents an isomorphism class of 
 virtual Quaternionic vector bundles over $A$.   
 
 On the other hand, we use KQ-cycles  (or generalize  KR-cycles) over $C^*$-algebras to model Majorana zero modes.  
 Although there is no local spectral 
 flows, we can still compute the Fredholm index, which is viewed as the analytical index of Majorana zero modes.

 Let $D$ be a skew-adjoint Fredholm operator  representing a Hamiltonian acting on the Hilbert space 
 $\mathcal{H}_\mathbb{H} = \mathcal{H}_\mathbb{C} \oplus \Theta \mathcal{H}_\mathbb{C} $. Based on the decomposition of the Hilbert space,
 we can decompose $D$ accordingly,
 $$
 D = \begin{pmatrix}
       0 & \Theta H \Theta^* \\
       H & 0
     \end{pmatrix} : \,\, \begin{matrix}
                      \mathcal{H}_\mathbb{C} \\
                         \oplus \\
                         \Theta \mathcal{H}_\mathbb{C}
                         \end{matrix} \rightarrow 
                             \begin{matrix}
                            \Theta  \mathcal{H}_\mathbb{C} \\
                             \oplus \\
                            \mathcal{H}_\mathbb{C}
                             \end{matrix}
 $$
 where $H$ is viewed as the Hamiltonian acting on $\mathcal{H}_\mathbb{C}$ and $\Theta H \Theta^*$ is the Hamiltonian acting on $\Theta \mathcal{H}_\mathbb{C}$.
 Obviously, $H$ and $\Theta H \Theta^*$  are similar, so that they have the same spectrum. Furthermore, $D$ is an unbounded operator if and only if
 $H$ is an unbounded operator, and $D$ is skew-adjoint, i.e., $D^* = - D$, if and only if $H$ satisfies the condition $H^* = - \Theta H \Theta^*$.

 By assumption, $D$ is a Fredholm operator, $\ker D$ is finite dimensional. So the analytical index of Majorana zero modes can be defined by
 \begin{equation}
    ind_a(D) = \dim \ker D \quad \text{mod 2}
 \end{equation}

 Define the real structure   
 $  J = \begin{pmatrix}
         0 & \Theta^* \\
         \Theta & 0
      \end{pmatrix} $ such that $J = J^*$ and $J^2 = 1$. Finally, the grading operator 
      $\gamma = \begin{pmatrix} 
                  1 & 0  \\
                  0 & -1
                \end{pmatrix} $ is added by hand.
In sum, the quintuple $(A, \mathcal{H} \oplus \Theta \mathcal{H}, D, J, \gamma )$ defines an even KQ-cycle
such that 
$$
 JD = - DJ, \quad J^2 = 1, \quad J\gamma = - \gamma J
$$
which is viewed as a $KR_2$-cycle modeling Majorana zero modes.
 
 \begin{defn}
  For a 2-summable KQ-cycle $(A, \mathcal{H} \oplus \Theta \mathcal{H}, D, J, \gamma )$ and a projection $p \in A$ representing
  a class $[p] \in KR_{2}(A)$ that generates a $\mathbb{Z}_2$ component, 
  the 2d noncommutative topological $\mathbb{Z}_2$ invariant is defined by the topological index,
  \begin{equation}
      ind(p) = \langle Ch^0(D, \gamma), Ch_0(p) \rangle =  \frac{1}{2 \pi}Tr(\gamma p [F, p][F, p]) 
 \end{equation}
  where $F$ is defined by functional calculus $  F := D |D|^{-{1}}$.
 \end{defn}

 In other words, $ind(p)$ is the Fredholm index of $pDp|_{p\mathcal{H}}: p\mathcal{H} \rightarrow p \Theta \mathcal{H}$, 
 $$
 pDp = \begin{pmatrix}
       0 & \Theta pHp \Theta^* \\
       pHp & 0
     \end{pmatrix} : \,\, \begin{matrix}
                      p\mathcal{H}_\mathbb{C} \\
                         \oplus \\
                         \Theta p\mathcal{H}_\mathbb{C}
                         \end{matrix} \rightarrow 
                             \begin{matrix}
                            \Theta p \mathcal{H}_\mathbb{C} \\
                             \oplus \\
                           p \mathcal{H}_\mathbb{C}
                             \end{matrix}
 $$
 where the projection $p$ satisfies the compatibility condition $p \Theta = \Theta p$, that is, 
 $p$ is the identity on $\Theta \mathcal{H}$.
 As a notation, we write $\bar{H} = \Theta H \Theta^*$. 
 Furthermore, if we  write $F$ as
 $$
 F = {D}( D^2)^{-1/2} = \begin{pmatrix}
                                    0 & \bar{ H}( H \bar{H} )^{-1/2} \\
                          H ( \bar{H}H )^{-1/2}     & 0      
                                \end{pmatrix}
 $$
 then the Fredholm index $ind(p)$ is given by (up to a normalization constant) 
 $$
 ind(p) = p[\bar{H}(H\bar{H})^{-1/2}, p ] [H(\bar{H}H)^{-1/2}, p ] - p [H(\bar{H}H)^{-1/2}, p ] [\bar{H}(H\bar{H})^{-1/2}, p ]
 $$
 The above formula has a clear meaning of spectral flow as in the Fredholm index of a Dirac operator. When we 
 consider the parity of Majorana zero modes, we take the mod 2 spectral flow as $ind(p)$ modulo 2. However,
 if the projection $p$ generates a $\mathbb{Z}_2$ component in $KR_2(A)$, then $ind(p)$ is inherently $\mathbb{Z}_2$-valued, so
 in the definition of $ind(p)$ it is redundant to add ``mod 2''.
 
 \begin{defn}
  For a 3-summable $KQ$-cycle
 $(A, \mathcal{H}_\mathbb{H}, D, J) $ and a unitary operator $u \in A$ representing a class $[u] \in KR_6(A)$ that generates a $\mathbb{Z}_2$ component,
 the 3d noncommutative topological $\mathbb{Z}_2$ invariant is defined as the topological index,
  \begin{equation}
   ind(PuP) = \langle Ch^1(D), Ch_1(u) \rangle = \frac{1}{4 \pi^2}Tr[( u^{-1} [F, u])^3] 
 \end{equation}
where the sign of $D$ is $F = D|D|^{-1/2}$ and the projection $P = (1-F)/2$.
 
 \end{defn}

Since in a topological insulator, we are interested in the occupied bands below the Fermi energy, which is assumed to be located at the zero energy level, we define 
the projection operator $P$ as the negative projection.
And by the assumption that  $[u] \in KR_6(A)$ is a  $\mathbb{Z}_2$ generator in the topological KR-theory of $A$, it is not necessary 
to add ``mod 2'' in the definition of $ind(PuP)$.

Finally, we use the local index formula to define the topological index in terms of residue traces.
\begin{defn}
For a 2-summable KQ-cycle $(A, \mathcal{H} \oplus \Theta \mathcal{H}, D, J, \gamma )$ and a projection $p \in A$ representing
  a class $[p] \in KR_{2}(A)$ that generates a $\mathbb{Z}_2$ component, 
  the 2d noncommutative topological $\mathbb{Z}_2$ invariant is defined by the topological index,
\begin{equation}
   ind(p) = \langle Res_{s=0}\Psi(D, \gamma), Ch_0(p) \rangle = \frac{1}{2}\cutint \, \gamma p [D,p][D, p] |{D}|^{-2} 
\end{equation}
\end{defn}

\begin{defn}
  For a 3-summable $KQ$-cycle
 $(A, \mathcal{H}_\mathbb{H}, D, J) $ and a unitary operator $u \in A$ representing a class $[u] \in KR_6(A)$ that generates a $\mathbb{Z}_2$ component,
 the 3d noncommutative topological $\mathbb{Z}_2$ invariant is defined as the topological index,
\begin{equation}
        ind(PuP) =  \langle Res_{s=0}\Psi(D), Ch_1(u) \rangle 
                 = \phi_1(u^{-1}, u)- \phi_3(u^{-1}, u, u^{-1}, u)  
\end{equation}
where $(\phi_1, \phi_3)$ is the $(b, B)$-cocycle defined by \eqref{cycphi1} and \eqref{cycphi3}.
\end{defn}

It is convenient to compute the noncommutative topological $\mathbb{Z}_2$ invariant using the residue traces if we know 
the spectrum of the Dirac operator $D$, which happens for example in isospectral deformation problems.

\section{Kane--Mele invariant}

In this section, we will first reformulate the Kane--Mele invariant as a geometric object
that compares the orientations of determinant and Pfaffian line bundles. Second, we introduce the fixed point $C^*$-algebra of
time reversal symmetry in the noncommutative 2-torus $\mathbb{T}^2_\theta$ and the quantum 3-sphere $\mathbb{S}^3_\theta$ \cite{L13}.
Finally, we define the noncommutative Kane--Mele invariant based on Quillen's construction of determinant line bundle over Fredholm operators.   

\subsection{Geometric Kane--Mele invariant}

The topological index of  Majorana zero modes is basically the Chern character, which can be viewed as the action functional 
of a topological insulator. One step further, if one considers the effective quantum field theory, then an equivalent topological invariant can be defined,
which is the Kane--Mele invariant  and basically it is the exponentiated topological $\mathbb{Z}_2$ index. 
The equivalence between the mod 2
topological index and the Kane--Mele invariant was proved in \cite{FM13, KLW1501}. From a different perspective, the Kane--Mele invariant gives rise to a fully 
extended topological quantum field theory \cite{KLW16}. In this paper, we will focus on the geometric picture of the Kane--Mele invariant 
with the help of determinant and Pfaffian line bundles.

Now let us first recall the definition of the Kane--Mele invariant \cite{KM05},
which works in all dimensions.
From the topological band theory of a topological insulator, its band structure is modeled by the Bloch bundle $\pi: (\mathcal{B}, \Theta) \rightarrow (X, \tau)$.
One always assumes that the band structure is non-degenerate, that is, the Bloch bundle has the decomposition $\mathcal{B} = \oplus_{i =1}^N \mathcal{B}_i$ if it has
$N $ occupied bands, and each sub-bundle $\mathcal{B}_i$ is of rank 2. The characteristic feature of a topological insulator is completely determined by the top band, i.e., the top sub-bundle $\mathcal{B}_N$, we call it
the Hilbert bundle from now on, denoted by $\mathcal{H} = \mathcal{B}_N$. 

The rank 2 Hilbert bundle $\pi: \mathcal{H} \rightarrow X$   is characterized by the transition function
$w: X \rightarrow U(2)$.
By the $\mathbb{Z}_2$-equivariant CW complex structure of the momentum space $(X, \tau)$, $w$ is entirely determined 
by the fixed points $X^\tau$. So the Hilbert bundle is obtained from the trivial bundle $\pi: X \times \mathbb{C}^2 \rightarrow X$ by twisting around $X^\tau$,
and the twists are kept track of by $w: X^\tau \rightarrow U(2)$. One important property of $w$ is that it is a skew-symmetric matrix at $X^\tau$,
i.e., $w^T(x) = - w(x)$ for any $x \in X^\tau$. 

The Kane--Mele invariant of a topological insulator is defined by,
 \begin{equation} \label{KMinv}
   \nu : =  \prod_{x \in X^\tau} sgn(pf [w(x)]) = \prod_{x \in X^\tau} \frac{pf [w(x)]}{\sqrt{\det [w(x)]}}   
  \end{equation} 
where $X^\tau$ is the set of fixed points, and $w: X \rightarrow U(2)$ is the transition function of the Hilbert bundle $\pi: \mathcal{H} \rightarrow X$. 
Since the transition function $w$ becomes  skew-symmetric at $X^\tau$,
it makes sense to take the Pfaffian function. Thus the Kane--Mele invariant is defined as the product of the signs of Pfaffians over the fixed points,
recall that for a skew-symmetric matrix $A$ ($A^* = - A$) one has the relation $det(A) = [pf(A)]^2$.

Let us look into the geometric picture of the Kane--Mele invariant.
 The Hilbert bundle $\pi: \mathcal{H}   \rightarrow X $  is of rank 2, it can be further assumed to have a pair of almost global  sections 
 except for one fixed point,
 $$
 (\psi, \Theta \psi) \in \Gamma(X, \mathcal{H})
 $$
 with possible twists around $X^\tau$ recorded by $w$.
 Define its determinant line bundle   as  
$$
Det\, \mathcal{H} : = \wedge^{2} \mathcal{H} \rightarrow X
$$ which is a complex line bundle with a global section
$$
\psi \wedge \Theta \psi \in \Gamma(X, Det \,\mathcal{H})
$$ with possible twists at $X^\tau$ induced from $w$. Notice that the determinant line bundle is trivializable, since 
the first Chern class $c_1(Det \mathcal{H}) = 0$.

The time reversal operator  $\Theta$ on $\mathcal{H}$ induces an action $ \det \Theta$   on $Det\, \mathcal{H}$. 
  More precisely, for the global  section 
    $\psi \wedge \Theta \psi \in \Gamma(X, Det\, \mathcal{H})$, the action of $\det \Theta$  is defined by
   \begin{equation} \label{RealCond}
          \det \Theta  \cdot (\psi \wedge \Theta \psi) = \Theta \psi \wedge \Theta^2 \psi  =  \Theta \psi \wedge (- \psi) =  \psi \wedge \Theta \psi 
   \end{equation}
  In other words, $\psi \wedge \Theta \psi$ is a real section with respect to $\det \Theta$. 
  Since $\det \Theta$ is a real structure such that $(\det \Theta)^2 = 1$, the determinant line bundle 
  $\pi: (Det \, \mathcal{H}, \det \Theta) \rightarrow (X, \tau)$ is a Real vector bundle. 
  Indeed, from the real condition \eqref{RealCond} $Det \, \mathcal{H}$ is a Real vector bundle
  such that $\det \Theta\, Det \, \mathcal{H} \cong Det \, \mathcal{H} $.

  It is well-known that the first Chern class of $\mathcal{H}$ can be represented by the determinant line bundle as an element
  in the Picard group, which is isomorphic to the 2nd cohomology group $H^2(X, \mathbb{Z})$, 
  $$
  c_1(\mathcal{H}) = [Det \, \mathcal{H}] \in H^2(X, \mathbb{Z})
  $$
  It was proved in \cite{KLW1501} that the Hilbert bundle is isomorphic 
  to its conjugate bundle $\overline{\mathcal{H}}$ or its dual bundle $\mathcal{H}^*$, 
  $$
  \mathcal{H} \cong \overline{\mathcal{H}} \cong \mathcal{H}^* 
  $$
  So we have  $Det \, \mathcal{H} \cong Det  \, \mathcal{H}^*$,  
  and then 
  $$
  2 c_1(\mathcal{H}) =[Det\, \mathcal{H}] + [Det\, \mathcal{H}^*] = 0
  $$ This gives another proof of the fact that the first Chern class 
  of the Hilbert bundle $\mathcal{H}$ is a 2 torsion.
   
  With the real condition $\det \Theta\, (\psi \wedge \Theta \psi) = \psi \wedge \Theta \psi$,  $\psi \wedge \Theta \psi$ is called a Majorana state, and 
  $\psi$ or $\Theta \psi$ is called a chiral state.
  Note that the chiral states in a Majorana state have the same value $\psi(x)  =  \Theta \psi(x) $ only at some fixed point $x \in X^\tau$.     
  $X^\tau$ is the set of real points with respect to $\tau$, 
  and the restriction of a chiral state, i.e., $\psi|_{X^\tau}$ or $\Theta \psi|_{X^\tau}$, is a real state with respect to $ \Theta$.
  In other words,  a restricted chiral state is a   section of  a \emph{real} line bundle over the real points $X^\tau$.  
   This \emph{real} line bundle is called the Pfaffian line bundle  
  $(Pf,  \Theta) \rightarrow X^\tau$, or simply $Pf \rightarrow X^\tau$, since $ \Theta$ is a quaternionic structure only on the fixed points $X^\tau$. 
  This is the genuine Pfaffian line bundle defined by real states over real points with respect to $\Theta$ and $\tau$ respectively.
  
   The structure group of the above Pfaffian line bundle $\pi: Pf \rightarrow X^\tau$ is $GL(1, \mathbb{C}) = \mathbb{C}^\times$, 
   and it can be extended to the whole momentum space $X$,
  $$
   \pi: Pf \rightarrow X, \quad \text{with} \quad h: X^\tau \rightarrow \mathbb{C}^\times
  $$
  where $h$ is the transition function. In other words, the Pfaffian line bundle $\pi: Pf \rightarrow X$ is obtained from 
  the trivial line bundle $X \times \mathbb{C} \rightarrow X $ by twisting at $X^\tau$, and the twists are kept track of by $h: X^\tau \rightarrow \mathbb{C}^\times$.
  Alternatively, by taking one chiral state $\psi$ or $\Theta \psi$ from the Hilbert bundle $\mathcal{H}$, 
  we define the Pfaffian line bundle $\pi: Pf(\mathcal{H}) \rightarrow X$, or simply $\pi: Pf \rightarrow X$. 
  
  The relation between the Pfaffian line bundle  $ Pf(\mathcal{H}) \rightarrow X$ and
   the determinant line bundle $Det(\mathcal{H}) \rightarrow X$ can be established as follows. 
   There exists a canonical isomorphism,
   $$
   Pf(\mathcal{H}) \otimes Pf(\mathcal{H}) \cong Det(\mathcal{H}), \quad (\psi, \psi) \mapsto \psi \wedge \Theta \psi
   $$ 
   which is commonly written as $Pf^2 \cong Det$.
   Note that if we chose the other chiral state to define the Pfaffian line bundle, the determinant line bundle would be the same,
   $$
   Pf(\mathcal{H}) \otimes Pf(\mathcal{H}) \cong Det(\mathcal{H}), \quad (\Theta\psi, \Theta\psi) \mapsto \Theta\psi \wedge \Theta(\Theta \psi) = \psi \wedge \Theta \psi
   $$ 
   
   The relation $Pf^2 = Det$ can be further clarified by looking at the relevant structure groups and transition functions.
   The Hilbert bundle has the structure group $U(2)$, if we assume
   that the physical states are normalized such that $\langle \psi, \psi \rangle = 1$, then $U(2)$ is reduced to $SU(2)$. The Pfaffian line 
   bundle has the structure group $\mathbb{C}^\times$, which is reduced to $U(1)$ after normalization. The determinant line bundle is isomorphic 
   to the trivial bundle $Det(\mathcal{H}) \cong X \times \mathbb{C}$, whose structure group is trivial.
   One way to see this is for the ideal $I$ defining the exterior product,
   $$
   Det \, \mathcal{H} = \mathcal{H} \wedge \mathcal{H} = \mathcal{H} \otimes \mathcal{H} / I \cong \mathcal{H} \otimes \mathcal{H}^* / I \cong X \times \mathbb{C}
   $$
   because of the isomorphism $\mathcal{H} \cong \mathcal{H}^*$. By the relation between the Hilbert bundle and
   the Pfaffian line bundle,  their transition functions are connected by the Pfaffian function, i.e., $h = pf(w)$, at the fixed points $X^\tau$,
   $$
     pf: SU(2) \rightarrow U(1), \,\, w(x) =\begin{pmatrix}
                  0 & e^{-i\beta(x)} \\
                  -e^{i\beta(x)} & 0
                \end{pmatrix}  \mapsto  h(x)=e^{-i\beta(x)}, \,\, \forall x \in X^\tau
   $$
   for some real-valued function $\beta: X \rightarrow \mathbb{R}$. Furthermore, we can assume the function $\beta$ respects the time 
   reversal symmetry, and reduce both $SU(2)$ and $U(1)$ to $\mathbb{Z}_2$ \cite{KLW16}.
   
   \begin{prop}
     The geometric Kane--Mele invariant is given by the product of possible twists of the Pfaffian line bundle over the fixed points,
     \begin{equation}
   \nu = \prod_{x \in X^\tau} h(x) = \prod_{x \in X^\tau} pfaff(x)
\end{equation}
where $h: X^\tau \rightarrow \mathbb{Z}_2$ is the transition function of $\pi: Pf \rightarrow X^\tau$, and $pfaff$ is the global section of
$\pi: Pf \rightarrow X^\tau$ representing its orientation.
   \end{prop}
  \begin{proof}
   
   There exists a canonical orientation 
   of the Pfaffian line bundle represented by a global section $pfaff \in \Gamma(X^\tau, Pf)$ 
   since $Pf \cong X^\tau \times \mathbb{C}$ is a trivial \emph{real} bundle. Here ``real'' means
   complex with respect to the quaternionic structure $\Theta$, so the Pfaffian line bundle can be 
   viewed as $Pf \cong X^\tau \times \mathbb{R}$ with possible twists recorded by 
   $$
   h: X^\tau \rightarrow \mathbb{Z}_2
   $$
   In addition, the global section $pfaff$ can be identified with $h$, since $pfaff$ has the form,
   $$
   pfaff: X^\tau \rightarrow X^\tau \times \mathbb{Z}_2
   $$
   
   After applying the real condition $\det \Theta\, Det \, \mathcal{H} = Det \, \mathcal{H} $,
   the determinant line bundle is viewed as $Det \, \mathcal{H} \cong X \times \mathbb{R}$, so it
   has a real global section 
   $$
   det:X \rightarrow Det\, \mathcal{H} = X \times \mathbb{R} 
   $$ 
    Hence the Kane--Mele invariant is understood
as the ratio of those real global sections over the fixed points,
\begin{equation*}
   \nu =  \frac{pfaff}{det|_{X^\tau}} = \prod_{x \in X^\tau} h(x)  
\end{equation*}
where $h(x) = \pm 1$ is interpreted as the comparison of orientations represented by $pfaff$ and $det$ respectively.
In fact, the global section of the determinant line bundle is trivial, i.e., $det = 1$, so the Kane--Mele invariant is completely
determined by the global section of the Pfaffian line bundle. 
\end{proof}

\subsection{A family index theorem}

Before we generalize the above discussion about the geometric Kane--Mele invariant to noncommutative manifolds, let us recall the relation
between the  determinant (or Pfaffian) line bundle and a family index theorem. 
The first part of this subsection is closely following \cite{F87}.

Let $Y$ be a compact spin manifold, and $D: \mathcal{H}_+ \rightarrow \mathcal{H}_- $ be a Dirac operator mapping $\mathcal{H}_+$ to $\mathcal{H}_-$,
where $\mathcal{H}_+$ and $\mathcal{H}_-$ are finite dimensional Hilbert spaces. Taking the top exterior products of the Hilbert spaces, denoted by
$\det \mathcal{H}_\pm = \wedge^{max} \mathcal{H}_\pm$, then $D$ induces a natural map between them,
$$
\det D : \det \mathcal{H}_+ \rightarrow \det \mathcal{H}_-
$$
In other words, $\det D$ can be viewed as an endomorphism,
$$
\det D \in  (\det \mathcal{H}_+)^* \otimes \det \mathcal{H}_-
$$

Next we promote the above construction on vector spaces to vector bundles. Let $\pi: Z \rightarrow X$ be a vector bundle 
with a typical fiber $Z_x \cong Y$, where each fiber is assumed to
be spin. Now let $D$ be a family of  Fredholm operators parametrized by the base manifold $X$, 
$$
D_x :  (\mathcal{H}_+)_x \rightarrow (\mathcal{H}_-)_x 
$$ 
In other words, $D$ is a bundle map between two Hilbert bundles,
 $$
   \xymatrix{
  \mathcal{H}_+ \ar[rd]  \ar@{-}[rr]^D & \ar[r] & \mathcal{H}_-  \ar[ld] \\
          &  X  & }
 $$
 Define the determinant line bundle of $D$ as a complex line bundle,
 $$
 \pi: Det D \rightarrow X
 $$
 such that each fiber 
 $$
 (DetD)_x = (\det \ker D_x)^* \otimes \det \ker D_x, \quad x \in X
 $$
 where the Hilbert spaces have been identified as
 $$
   (\mathcal{H}_+)_x \cong  (\ker D_x)^*, \quad  (\mathcal{H}_-)_x \cong \ker D_x
 $$

For each $x \in X$, the Fredholm operator $D_x$ has finite dimensional kernel and cokernel, and 
the Fredholm index is the analytical index,
$$
index(D_x) = \dim ker \, D_x - \dim coker \, D_x
$$
Therefore, one defines the index bundle as the formal difference bundle,
$$
\pi: Ind(D) \rightarrow X, \quad Ind (D) = ker D - coker D
$$
where $kerD$ and $coker D$ are vector bundles over $X$. 
Obviously, the index bundle of $D$ gives an element in the K-theory of the base manifold,
$$
[Ind(D)] \in K(X)
$$
Furthermore, the determinant line bundle of the index bundle $\pi: Ind(D) \rightarrow X$ is isomorphic 
to the determinant line bundle of the family of self-adjoint Fredholm operators $D$,
$$
det(Ind(D)) \cong Det(D)
$$
As a consequence,
 the first Chern class of the index bundle is represented by the determinant line bundle of $D$.
$$
c_1(Ind(D)) = [DetD] \in H^2(X, \mathbb{Z})
$$

If the typical fiber $Y$ is spin and $(8k +2)$-dimensional, then there exists a Pfaffian line bundle, 
$$
\pi: Pf(D) \rightarrow X
$$
such that 
$$
Pf(D) \otimes Pf(D) \cong Det(D)
$$
There exists a push-forward map,  called the Gysin map, between the real K-groups of the total and base spaces in $\pi: Z \rightarrow X$,
$$
\pi_!: KO(Z) \rightarrow KO^{-2}(X)
$$
If one considers the index bundle $\pi: Ind(D) \rightarrow X$ in this real case, then one has the identity,
$$
 \pi_!([1]) = [Ind(D)] \in KO^{-2}(X)
$$
where $[1]$ is a generator in $KO(Z)$ \cite{F87}. 

Moreover, by the Chern--Weil theory, there exists a map induced by the Pfaffian form, 
$$
Pfaff: KO^{-2}(X) \rightarrow H^2(X, \mathbb{Z})
$$
where $H^2(X, \mathbb{Z})$ is the 2nd (de-Rham) cohomology group characterizing real line bundles over $X$.
By a theorem in \cite{F87}, one has the relation,
$$
c_1(Pf(D)) = Pfaff([Ind(D)]) \in H^2(X, \mathbb{Z})
$$

Now we apply the above construction to the effective Hamiltonian of Majorana zero modes,
which is a family of skew-adjoint Fredholm operators, 
$$
D(x) = \begin{pmatrix}
                                                 0 & \Theta H(x) \Theta^* \\
                                                 H(x) & 0
                                                \end{pmatrix}
$$
 We define the determinant line bundle of $D$ as
\begin{equation}
   Det(D) =  (\det ker \, H )^* \otimes \det ker \, \Theta H \Theta^{*} 
\end{equation}
Suppose $ker H$ and $ker \Theta H \Theta^{*} $ are 1 dimensional, in this case, 
$$
Det(D) \cong Det(\mathcal{H})
$$
Since the first Chern class of the Hilbert bundle $\pi: \mathcal{H} \rightarrow X$ is a 2-torsion, 
i.e., $2c_1(\mathcal{H}) = 0$,
there exists a Pfaffian line bundle $\pi: Pf(D) \rightarrow X$.

In addition, $H$ and $\Theta H \Theta^{*}  $ have zero modes only at the fixed points $X^\tau$, 
the determinant line bundle $Det(D)$ is viewed as
a vector bundle over $X^\tau$,
$$
\pi: Det(D) \rightarrow X^\tau
$$
And the Pfaffian line bundle is then a real line bundle,
$$
\pi: Pf(D) \rightarrow X^\tau
$$
such that 
$$
Pf(D) \otimes Pf(D) \cong Det(D)
$$
Therefore, the skew-adjoint operator $D$ has the family index, 
$$
[Ind(D)] \in KO^{-2}(X^\tau) 
$$
and the first Chern class is 
$$
c_1(Pf(D)) = Pfaff([Ind(D)]) \in H^2(X^\tau, \mathbb{Z}_2) 
$$

In order to get the topological index from the above family index, one has to project the set of fixed points to a point $p: X^\tau \rightarrow pt$,
which induces the push-forward map in real K-theory $p_*: KO^{-2}(X^\tau) \rightarrow KO^{-2}(pt)$. In \cite{KLW1501}, we realized the abstract point $\{pt\}$
as a specific fixed point with top codimension, which is viewed as the effective boundary in geometry.

\subsection{Fixed point algebra}
 
We will use the noncommutative 2-torus as a concrete example, 
let us first recall its definition and introduce the time reversal symmetry in this context.
We will define the fixed point algebra of the time reversal symmetry, which is the noncommutative analog of the set of fixed points.

In the classical case, if the 2-torus $\mathbb{T}^2$ is parametrized by the angles
   $$
   \mathbb{T}^2 = \{ (e^{i\theta_1}, e^{i\theta_2}) ~|~ \theta_i \in (-\pi, \pi ], i = 1, 2 \}
   $$
   then the time reversal transformation on $\mathbb{T}^2$ is defined by 
   $$
   \tau:  \mathbb{T}^2 \rightarrow \mathbb{T}^2; \quad (e^{i\theta_1}, e^{i\theta_2}) \mapsto (e^{-i\theta_1}, e^{-i\theta_2})
   $$
So $(\mathbb{T}^2, \tau)$ is a Real space with $\tau^2 = 1$.
Let $C(\mathbb{T}^2)$ denote the continuous functions on $\mathbb{T}^2$, if the complex conjugation is denoted by $K$,
$$
K: C(\mathbb{T}^2) \rightarrow C(\mathbb{T}^2), \quad K(f) = \bar{f}
$$
then we view $(C(\mathbb{T}^2), K)$ as a Real function algebra with $K^2 = 1$. The invariant subspace of $C(\mathbb{T}^2)$ with respect to $\tau$ and $K$ is defined as,
$$
C_{R} (\mathbb{T}^2):= \{ f\in C(\mathbb{T}^2) ~|~  Kf(x) = \overline{f(x)} = f(\tau(x))  \}
$$
The time reversal operator $\Theta$ is an anti-unitary operator such that $\Theta^2 = -1$.
For example, if we define the time reversal operator by 
$$
\Theta := i\sigma_2K =\begin{pmatrix}
                       0 & K \\
                       -K & 0
                     \end{pmatrix}
$$
where $\sigma_2$ is the 2nd Pauli matrix, then  $\Theta$ acts on $C_R(\mathbb{T}^2)$ by
$$
 \Theta f = \begin{pmatrix}
                       0 & K \\
                       -K & 0
                     \end{pmatrix} \begin{pmatrix}
                                     f(x) \\
                                     \tau^*f(x)
                                     \end{pmatrix} = \begin{pmatrix}
                                                       Kf(\tau(x)) \\
                                                       -Kf(x)
                                                       \end{pmatrix} = \begin{pmatrix}
                                                                          f(x) \\
                                                                          -\tau^*f(x)
                                                                          \end{pmatrix}
$$
where $\tau^*f(x) = f(\tau(x))$ and $C_R(\mathbb{T}^2) $ is assumed to be diagonally embedded into 
$(1\oplus \tau^*)C_R(\mathbb{T}^2)$ first.
This is the case for the commutative Real $C^*$-algebra $C_R(\mathbb{T}^2) $.

Let $\theta \in \mathbb{R} \setminus \mathbb{Q} $ be an irrational real parameter,
     the noncommutative 2-torus $\mathbb{T}_\theta^2 = C^*(u,v)$ is the universal $C^*$-algebra generated by two unitary operators
 $u,v$  satisfying $uv = e^{2 \pi i \theta} vu$. $\mathbb{T}_\theta^2$ is viewed as the function algebra of the underlying noncommutative manifold. 
 Again, the time reversal operator $\Theta$ is defined as an anti-unitary operator acting on $\mathbb{T}_\theta^2$ such that $\Theta^2 = -1$.
 For example, $\Theta$ can be similarly defined as,
 \begin{equation}
    \Theta := \begin{pmatrix}
                 0 & * \\
                 -* & 0
              \end{pmatrix}
  \end{equation}
where $*(a) = a^*$ is the $*$-operation of the $C^*$-algebra $\mathbb{T}_\theta^2$.
 
 As a concrete example, the abstract noncommutative 2-torus can be realized on a square lattice model. 
 Let  $\mathcal{H} = \ell^2(\mathbb{Z}^2)$ be the Hilbert space, 
 for any $\xi \in \ell^2(\mathbb{Z}^2)$, the unitary operators $u$ and $v$ can be explicitly defined as 
 $$
 u\xi(n,m) = e^{ \pi i m\theta}   \xi(n-1,m), \,\,  v \xi(n,m) = e^{ -\pi i n\theta}   \xi(n,m-1)
 $$ 
 Sometimes $u$ and $v$ are viewed as a pair of conjugate variables, so 
 $$
 \Theta \begin{pmatrix}
          u \\
          v
        \end{pmatrix} = \begin{pmatrix}
                            0 & * \\
                            - * & 0
                         \end{pmatrix} \begin{pmatrix}
                                             u \\
                                             v
                                             \end{pmatrix}  = \begin{pmatrix}
                                                                 v^* \\
                                                                 -u^*
                                                                 \end{pmatrix}
 $$
So the effective result of the action by $\Theta$ is, 
    $$
 \Theta u\xi(n,m) = e^{ \pi i n\theta}   \xi(n,m+1), \,\, \Theta v \xi(n,m) = - e^{ -\pi i m\theta}   \xi(n+1,m)
 $$
 
 \begin{defn}
    Define the time reversal operator $\Theta = i\sigma_2*$ acting on the noncommutative 2-torus, or explicitly,
 $$
  \Theta(u) = v^*, \quad \Theta(v) = - u^*, \quad \Theta(u^*) = v, \quad \Theta(v^*) = -u
 $$
 \end{defn}

 The action of $\Theta$ respects the order of a product since we are dealing with a fermionic system, for example,
 $$
 \Theta(uv) = \Theta(u) \Theta(v) = v^*(-u^*) = -v^*u^*
 $$
 As a consequence, this action is compatible with the relation $uv = e^{2\pi i \theta} vu$.
 
 If we view $(u, v, u^*, v^*)$ as a four-band system, then the time reversal operator is represented by 
 $$
 \Theta = \sigma_0 \otimes i\sigma_2 *= \sigma_1 \otimes i\sigma_2 = \begin{pmatrix}
                                         0 & i\sigma_2 \\
                                         i\sigma_2 & 0 
                                       \end{pmatrix}
 $$
 where $i\sigma_2$ is a skew-symmetric matrix such that $(i\sigma_2)^2 = -\sigma_0= -I_2$.
 
 \begin{defn}
    The fixed point algebra of the time reversal operator $\Theta$, denoted by $(\mathbb{T}^2_\theta)^\Theta$, is the $C^*$-algebra generated by
 $$
  x= e^{\pi i \theta} i uv^*, \quad y = e^{-\pi i \theta}i vu^*
 $$
 \end{defn}

 \begin{prop}
    The fixed point algebra of the time reversal operator $\Theta$ in $\mathbb{T}^2_\theta$ is given by
 \begin{equation}
    (\mathbb{T}^2_\theta)^\Theta = \langle x = e^{\pi i \theta} i uv^* ~|~ xx^*=x^*x = 1 \rangle \cong \mathbb{C}x
 \end{equation}
 \end{prop}
 \begin{proof}
    One has the following relations between the generators,
 $$
  yx =xy = -1,\quad x^* = -y, \quad  xx^*=x^*x = 1
 $$
 \end{proof}

 \begin{cor}
     The real K-theory of the fixed point algebra $(\mathbb{T}^2_\theta)^\Theta $  is given by
     \begin{equation}
        KO_i[(\mathbb{T}^2_\theta)^\Theta] \cong KO^{-i}(\mathbb{R}) = KO^{-i-1}(pt)
     \end{equation}
 In particular, we have
\begin{equation}
    KO_0[(\mathbb{T}^2_\theta)^\Theta] = \mathbb{Z}_2, \quad  KO_{1}[(\mathbb{T}^2_\theta)^\Theta] = \mathbb{Z}_2 
\end{equation}
 \end{cor}
 
 \begin{proof}
   
 Due to the general real structure $\Theta$ such that $\Theta^2 = -1$, the fixed point algebra is again viewed as a real algebra, 
 i.e.,  $ (\mathbb{T}^2_\theta)^\Theta  \cong \mathbb{R}x$.
  
 \end{proof}

Let us look into the geometry of the fixed point algebra $(\mathbb{T}^2_\theta)^\Theta$.
From the concrete form of the generator $x = e^{\pi i \theta} i uv^*$, it creates an imaginary (due to the imaginary unit ``i'') channel 
between $u$ and $v^*$. If we view $uv =  e^{2\pi i \theta} vu$ as a braiding rule, say to rotate $u$ about $v$, then one gets a phase $ e^{2\pi i \theta}$ after a twist. 
So $ e^{\pi i \theta}$ in $x$ is viewed as a half-twist. The same analysis applies to the other generator $y = e^{-\pi i \theta}i vu^*$,
which creates the opposite channel from $v$ to $u^*$ in the four-band system $(u,v, u^*, v^*)$. However, one channel is enough to characterize
the geometry since the opposite channel is easily recovered by $x^* = -y$.

 In 3 dimensions, the easiest example is given by the quantum 3-sphere $\mathbb{S}^3_\theta$  due to its close relation to the noncommutative 2-torus $\mathbb{T}^2_\theta$ \cite{L13}.
 $\mathbb{S}^3_\theta$ is generated by two operators,
 $$
 \alpha = \cos(\eta) u, \quad \beta = \sin(\eta) v, \quad \text{s.t.} \quad \alpha \alpha^* + \beta \beta^* =1
 $$ where
 $u,v$ are generators of $\mathbb{T}^2_\theta$ and the angle $\eta \in [0, \pi/2]$. 
 The time reversal operator in this case is again defined by 
 $$
 \Theta = i\sigma_2 *
 $$ where 
 $*$ is the adjoint-operation in $C^*$-algebras.  
 $\mathbb{S}^3_\theta$ also has a 1d fixed point algebra, 
 $$
 (S^3_\theta)^\Theta \cong \mathbb{R}x, \quad  x = e^{\pi i \theta} i \alpha \beta^*
 $$ 
 In other words, one also has a four-band system in $\mathbb{S}^3_\theta$ consisting of $(\alpha, \beta, \alpha^*, \beta^*)$.

 \subsection{Noncommutative Kane--Mele invariant}
 
 The Kane--Mele invariant \eqref{KMinv} can be understood geometrically by comparing the orientations of
 Pfaffian and determinant line bundles \cite{FM13, KLW15}. 
 We will generalize this geometric picture of the Kane--Mele invariant to the noncommutative 2-torus in this subsection. 
 In addition, this noncommutative Kane--Mele invariant can be interpreted as
 a mod 2 analytical index by the mod 2 spectral flow as before \cite{KLW15, KLW1501}. Finally, by a family index theorem of the Pfaffian line bundle over the
 fixed point algebra, the mod 2 analytical index is identified with the topological index.

  Quillen's construction \cite{Q85} of  the determinant line bundle on Riemann surfaces can be generalized to the noncommutative 2-torus \cite{CM09, FGK14}.
  Let $\mathscr{F}$ be the space of Fredholm operators over a separable infinite dimensional complex Hilbert space $\mathcal{H}$, 
   Quillen constructed the determinant line bundle $Det \rightarrow \mathscr{F} $ such that for any $T \in \mathscr{F}$,
     $$
      Det_T  = \wedge^{max}(ker\,T)^* \otimes \wedge^{max}(ker\,T)
     $$
     
   Assume that the self-adjoint Hamiltonian $H$ is parametrized by the noncommutative torus, that is, $H(a)$ is a family of Hamiltonians  
   for $a \in \mathbb{T}^2_\theta$. Hence the skew-adjoint Fredholm operator 
   $$
   D(a) = \begin{pmatrix}
                                                                                0 & \Theta H(a) \Theta^* \\
                                                                                H(a) & 0
                                                                               \end{pmatrix}
   $$ 
   is also depending on $a  \in \mathbb{T}^2_\theta $.
  
   Therefore, the determinant line bundle over the noncommutative 2-torus is defined as the pullback line bundle
   $\mathcal{L} := D^*Det$ such that
   $$
   \mathcal{L}_{a} = \wedge^{max}(ker H(a))^* \otimes \wedge^{max}(ker \Theta H(a)\Theta^*)
   $$
   $$
   \xymatrix{
\mathcal{L} \ar[d]^\pi \ar[r] & Det \ar[d]^\pi\\
\mathbb{T}^2_\theta \ar[r]^{D} &\mathscr{F}}
   $$

 Consider Fredholm operators acting on a real Hilbert space $\mathscr{F}_\mathbb{R} := \mathscr{F}(\mathcal{H}_\mathbb{R})$,
    which is a classifying space of KO-theory, i.e., $KO(X) \simeq [X, \mathscr{F}_\mathbb{R}]$.
    Similar to the determinant line bundle, we define the Pfaffian line bundle as $Pf \rightarrow \mathscr{F}_\mathbb{R}$ such that
   for any $T \in \mathscr{F}_\mathbb{R}$,      $ Pf_T  = \wedge^{max}ker\,T$.
  If we complexify the underlying real Hilbert space, then we have the relation between the Pfaffian and determinant line bundles,
     $$
     Pf \otimes Pf \cong Det
     $$
     
  Fix a real element $r \in (\mathbb{T}^2_\theta)^\Theta \cong \mathbb{R}x $, 
  the Hamiltonian $H(r)$ is a real Fredholm operator. 
  The Pfaffian line bundle over the fixed point algebra is defined similarly as the pullback line bundle  
   $\mathcal{L}^{1/2} := H^*Pf$ such that  
   $$
   \mathcal{L}^{1/2}_{r} = \wedge^{max}ker H(r)
   $$
   $$
   \xymatrix{
\mathcal{L}^{1/2} \ar[d]^\pi \ar[r] & Pf \ar[d]^\pi\\
(\mathbb{T}^2_\theta)^\Theta \ar[r]^{H} &\mathscr{F}_\mathbb{R}}
   $$
 As a consequence, when the determinant line bundle $\mathcal{L}$ is restricted to the fixed point algebra $(\mathbb{T}^2_\theta)^\Theta
 $, we obtain the desired relation 
 $$
 \mathcal{L}^{1/2} \otimes \mathcal{L}^{1/2} \cong \mathcal{L}|_{(\mathbb{T}^2_\theta)^\Theta}
 $$

    The  determinant line bundle $\mathcal{L}$ has a canonical trivialization $\sigma$ such that $\sigma = 1$ on the open subset of invertible operators in $\mathbb{T}^2_\theta$.
    Similarly there exists a trivialization $ \varrho$ on invertible operators in $(\mathbb{T}^2_\theta)^\Theta$,
    denoted by $(\mathbb{T}^2_\theta)^{\Theta, \times}\cong \mathbb{R}^\times x$, such that $\varrho^2 = 1 $.
    The fixed point algebra $(\mathbb{T}^2_\theta)^\Theta$ is generated by $x$ and $y$ with $x^* = -y$, so 
    the open set $(\mathbb{T}^2_\theta)^{\Theta, \times}\cong \mathbb{R}^\times x = \mathbb{R}^+x \sqcup \mathbb{R}^-x$ is represented
    by these two points $x \in \mathbb{R}^+x $ and $y \in \mathbb{R}^-x $.
     We only care about the orientation of the Pfaffian line bundle $\pi: \mathcal{L}^{1/2} \rightarrow (\mathbb{T}^2_\theta)^\Theta $, 
    so its structure group can be reduced from $\mathbb{R}^\times$ to $\mathbb{Z}_2$. Therefore, the global section $\varrho$ is essentially 
    determined by 
    $$
    \varrho : \{ x, y\} \rightarrow (\mathbb{Z}_2, \times) = \{ 1, -1\}
    $$
    
    \begin{defn}
       The noncommutative Kane--Mele invariant is defined by
    \begin{equation}
        \nu :=  \varrho(x) \varrho(y) \in \{ 1, -1\}
    \end{equation}
    where $x$ and $y$ are generators of the fixed point algebra $(\mathbb{T}^2_\theta)^\Theta$.
    \end{defn}

     If $\varrho(x)\varrho(y) = 1$, then
    the Pfaffian line bundle $\pi: \mathcal{L}^{1/2} \rightarrow (\mathbb{T}^2_\theta)^\Theta $ is orientable, that is, the Kane--Mele invariant
    $\nu$ is trivial. On the other hand, if $\varrho(x)\varrho(y) = -1$, then the Pfaffian line bundle is not orientable, that is,
    $\nu$ is non-trivial.
    
    \begin{prop}
            The Pfaffian line bundle $\pi: \mathcal{L}^{1/2} \rightarrow (\mathbb{T}^2_\theta)^\Theta$ is the $C^*$-version of 
            either the M{\"o}bius band $\pi: Mo \rightarrow S^1$ or the cylinder $S^1 \times \mathbb{R}$.

    \end{prop}
    \begin{proof}
       Since the generator $x$ of the fixed point algebra $(\mathbb{T}^2_\theta)^\Theta$ is a unitary operator, we have 
       $$
         (\mathbb{T}^2_\theta)^\Theta \cong C(S^1)
       $$
       which is not true for the fixed point algebra $(\mathbb{S}^3_\theta)^\Theta$. The Pfaffian line bundle
       $\pi: \mathcal{L}^{1/2} \rightarrow (\mathbb{T}^2_\theta)^\Theta$ is then the $C^*$-version of a real line bundle over the circle,
       which is either the  M{\"o}bius band or the cylinder according to its orientability.
    \end{proof}

   \begin{thm}
      The noncommutative Kane--Mele invariant and the noncommutative topological index are equivalent.
   \end{thm}
   \begin{proof}
     First the index bundle over the fixed point algebra,
     $$
     \pi: Ind(D) \rightarrow  (\mathbb{T}^2_\theta)^\Theta
     $$
     is defined by the Fredholm index of the skew-adjoint Fredholm operator $D$, i.e.,  the analytic index, 
     $$
       index(D) = \dim \ker D \,\,\text{(mod 2)}
     $$
     The fixed point algebra $(\mathbb{T}^2_\theta)^\Theta$ has two generators $x$ and $y$, whose geometry can be understood as a spectral flow
     between the generators of $\mathbb{T}^2_\theta$. For example, if the global section assigns $1$ (or trivial) 
     to  the generator $x = e^{\pi i \theta} i uv^*$, i.e., $\varrho(x) = 1$, then the local orientability of $\pi: \mathcal{L}^{1/2} \rightarrow (\mathbb{T}^2_\theta)^\Theta$
     around $x$ does not change, that is, there is no effective spectral flow at $x$. 
     Instead, if $\varrho(x) = -1$, then there exists a non-trivial spectral flow at $x$.
     In other words, the noncommutative Kane--Mele invariant $\nu =  \varrho(x) \varrho(y)$ is interpreted as
     a mod 2 analytic index using spectral flow. 
     
     Applying the Gysin map to the Pfaffian line bundle $\pi: \mathcal{L}^{1/2} \rightarrow (\mathbb{T}^2_\theta)^\Theta $,
      the index bundle gives an element in the real K-group of the fixed point algebra, 
        \begin{equation}
       [Ind(D)]  \in KO_1[(\mathbb{T}^2_\theta)^\Theta] = KO^{-2}(pt) 
    \end{equation}
     Note that the Pfaffian line bundle $\pi: \mathcal{L}^{1/2} \rightarrow (\mathbb{T}^2_\theta)^\Theta$ is a rank 1 real line bundle,
      so we know the topological index is in $ KO_1[(\mathbb{T}^2_\theta)^\Theta]  = \mathbb{Z}_2$. 
    This is a generalization of the family index theorem over the set of fixed points,  which computes the analytic index by the topological index.
    
    From the discussion about the topological index, we know the index pairing between the KQ-cycle and the Hilbert bundle also falls into $KO^{-2}(pt)$,
    so does the noncommutative topological index in \S \ref{TopInd}.
    Thus we establish the equivalence
    between the noncommutative Kane--Mele invariant over $(\mathbb{T}^2_\theta)^\Theta$ and the noncommutative topological index of the bulk $\mathbb{T}^2_\theta$.
  which is called the bulk-boundary correspondence for $\mathbb{T}^2_\theta$ generalizing the classical case \cite{KLW1501}.
  \end{proof}
 
 Similarly, it is easy to define the noncommutative Kane--Mele invariant on the noncommutative 3-sphere. The bulk-boundary correspondence for 
 3d noncommutative manifolds will be worked out in a future paper.

\nocite{*}
\bibliographystyle{plain}
\bibliography{NCZ2}

\end{document}